\documentclass[a4paper,onecolumn,10pt]{article}

\usepackage[utf8]{inputenc} 
\usepackage[T1]{fontenc}
\usepackage[english]{babel}
\usepackage{csquotes}
\usepackage[a4paper,total={160mm,250mm}]{geometry}
\usepackage[skip=2.5mm,indent=5mm]{parskip}
\usepackage{setspace}
\usepackage{amsmath}  
\usepackage{amsfonts}
\usepackage{amssymb}
\usepackage{amsthm}
\usepackage{hyperref}
\usepackage{graphicx}         
\usepackage[font=normalsize]{caption}
\usepackage{subcaption}
\usepackage{url}    
\usepackage{algorithm}
\usepackage{algpseudocode}
\usepackage{tikz}

\hypersetup{colorlinks=true,linktoc=all,linkcolor=black,citecolor=black,
		    filecolor=black,urlcolor=black}
		    
\newtheorem{theorem}{Theorem}[section]
\newtheorem{proposition}[theorem]{Proposition}
\newtheorem{lemma}[theorem]{Lemma}

\newcommand{\N}{\mathbb{N}}
\newcommand{\Z}{\mathbb{Z}}
\newcommand{\F}{\mathbb{F}}
\newcommand{\bH}{\mathbf{H}}
\newcommand{\bT}{\mathbf{T}}
\newcommand{\e}{\mathbf{e}}
\newcommand{\s}{\mathbf{s}}
\newcommand{\he}{\hat{\mathbf{e}}}
\newcommand{\wt}{\mathsf{wt}}
\newcommand{\wtr}{\mathsf{wt_r}}

\newcommand\blfootnote[1]{%
  \begin{NoHyper}
  \renewcommand\thefootnote{}\footnote{#1}%
  \addtocounter{footnote}{-1}%
  \end{NoHyper}}

\begin{document}

\title{Analysis of the Error-Correcting Radius of a Renormalisation Decoder for Kitaev’s Toric Code}

\author{Wouter Rozendaal\footnotemark[1] \footnotemark[2] \and Gilles Zémor\footnotemark[1] \footnotemark[3]}

\maketitle

\blfootnote{This work was presented in part at ISIT 2023, see \cite{RZ23}}
\renewcommand{\thefootnote}{\fnsymbol{footnote}}
\footnotetext[1]{Institut de Mathématiques de Bordeaux, UMR 5251, Université de Bordeaux, France} 
\footnotetext[3]{Institut universitaire de France}
\footnotetext[2]{Email: \textbf{wouter.rozendaal@math.u-bordeaux.fr}}\footnotetext[3]{Email: \textbf{gilles.zemor@math.u-bordeaux.fr}}
\renewcommand{\thefootnote}{\arabic{footnote}}

\begin{abstract}
Kitaev's toric code is arguably the most studied quantum code and is expected to be implemented in future generations of quantum computers. The renormalisation decoders introduced by Duclos-Cianci and Poulin exhibit one of the best trade-offs between efficiency and speed, but one question that was left open is how they handle worst-case or adversarial errors, i.e. what is the order of magnitude of the smallest weight of an error pattern that will be wrongly decoded. We initiate such a study involving a simple hard-decision and deterministic version of a renormalisation decoder. We exhibit an uncorrectable error pattern whose weight scales like $d^{1/2}$ and prove that the decoder corrects all error patterns of weight less than $\frac{5}{6}  d^{\, \log_{2}(6/5)}$, where $d$ is the minimum distance of the toric code.
\end{abstract}

\section{Introduction}
\label{sec:introduction}

Quantum computers are of major interest as they are capable of efficiently solving certain computational problems that are considered difficult with classical computers. However, the physical implementation of a quantum computer still remains a problem due to decoherence errors. An important challenge in the realisation of a full-scale quantum computer therefore consists in finding ways to protect quantum information from errors. 

So-called topological codes, and in particular Kitaev's toric code \cite{K03, DKLP02}, are expected to be at the core of the first generation of quantum computers that will incorporate error protection. Kitaev's code is one of the oldest quantum error-correcting codes, and by far the most studied one. Although it protects only two logical qubits, its appeal comes in part from its simple structure, its planar layout useful for physical implementations, and parity-check operators of low weight that only involve neighbouring qubits. Indeed, reliable syndrome measurements need low-weight check operators.

Quantum error-correcting codes need a fast decoder that will process the classical information obtained from quantum syndrome measurements, so as to be able to regularly put arrays of qubits back into their intended states. To make full use of the toric code, one therefore also requires an efficient decoding scheme. Many decoding algorithms have been put forward, for a comprehensive list see for example the introduction of \cite{DN21}.

The Minimum Weight Perfect Matching algorithm \cite{DKLP02}, based upon Edmonds' blossom algorithm \cite{E65}, is the standard reference since it outputs optimal solutions on bit-flip channels. However, it has a high time-complexity of $O(n^{3})$ \cite{K09}, where $n$ is the length of the toric quantum code. At the other end, the recent union-find decoder of Delfosse and Nickerson \cite{DN21} is in quasi-linear time and is a remarkable compromise between performance and speed. However, the decoder of \cite{DN21} does not lend itself to parallelisation, and in that respect, the decoder that stands out is arguably the renormalisation decoder of Duclos-Cianci and Poulin \cite{DP10[1], DP10[2]} that can be made to run in time $O(\log_2 n)$, with overall complexity $O(n \log_2 n)$. 

The focus of the present paper is the renormalisation decoder of Duclos-Cianci and Poulin, or rather decoders, since the renormalisation idea is quite versatile and lends itself to many variations. In the Kitaev code, qubits are indexed by the $n$ edges of a square tiling of a torus. The renormalisation idea stems from the fact that this square tiling contains a tower of $\frac{1}{2} \log_2 n$ subtilings, each isomorphic to the original tiling. Each subtiling is contained in the previous one and has four times fewer edges than the latter. Local transformations are then applied to the error pattern, so that the original decoding problem becomes an instance of the decoding problem on a Kitaev code of length $n/4$. A~decoder then applies this procedure recursively $\frac{1}{2} \log_2 n$ times.

Renormalisation decoders were tested in \cite{DP10[1], DP10[2]} over depolarising and bit-flip channels, and found to exhibit very good thresholds when combined with some message-passing techniques. One feature that remained a mystery however, is their behaviour over adversarial channels, i.e. their worst-case behaviour. Worst-case behaviour is not necessarily the prime feature of a decoder, but it is important nevertheless because it is difficult to establish a precise model for errors that occur during quantum computations and also because it governs the speed with which decoding error probabilities tend to zero in the low channel error regime.

In the present paper we study this worst-case behaviour issue. To this end we introduce a simple and relatively natural version of a renormalisation decoder, that in particular does not use a priori error probabilities and acts deterministically, basing its decisions purely on local syndrome values. We find that the renormalisation decoder allows for ``fractal-like'' error patterns that are wrongly decoded and scale as $d^{1/2}$, where $d$ is the code's minimum distance.  We also prove a lower bound of the form $\frac{5}{6} d^{\, \log_{2}(6/5)}$ for the weight of an uncorrectable error pattern. Finally, we argue that the sub-linear behaviour of the minimum distance will be a feature of any renormalisation decoder. 

The paper is organised as follows. Section~\ref{sec:toric code} recalls the construction of Kitaev's toric quantum code and its decoding. Section~\ref{sec:renormalisation decoder} describes the renormalisation decoder that will be studied throughout this paper. Section~\ref{sec:bit-flip channel} analyses the decoder's performance on the bit-flip channel and gives a numerically established threshold probability. Section~\ref{sec:adversarial channel} is the core of the paper, bounding the error-correcting radius of the renormalisation decoder. Section~\ref{sec:fractal errors} gives fractal-like wrongly decoded error patterns for more general deterministic renormalisation decoders.

\section{Kitaev's toric quantum code}
\label{sec:toric code}

An important class of quantum error-correcting codes is the class of CSS codes introduced independently by Calderbank-Shor and Steane \cite{CS96, S96}. A CSS code can be described by two binary matrices $\bH_X$ and $\bH_Z$ of size $r_X \times n$ and $r_Z \times n$ respectively, where $r_X, r_Z$ and $n$ are positive integers. The row spaces of $\bH_X$ and $\bH_Z$ should be orthogonal, in other words $\bH_X  \bH_Z^\intercal = \mathbf{0}$. The matrices $\bH_X$ and $\bH_Z$ should be thought of as the parity-check matrices of two binary codes $C_X$ and $C_Z$ satisfying $C_Z^\perp \subset C_X$. The length of the quantum code, i.e. the number of physical qubits of the code, equals $n$. It's dimension, i.e. the number of logical qubits, equals~$n - \mathrm{rank} \, \bH_X - \mathrm{rank} \, \bH_Z$.  The minimum distance $d$ of the code is equal to $\min(d_X,d_Z)$, where $d_X$ (resp. $d_Z$) is equal to the smallest Hamming weight of a non-zero codeword of $C_X$ (resp. $C_Z)$ that is not in the row space of $\bH_Z$ (resp. $\bH_X$). 

The toric code \cite{K03, DKLP02} proposed by Kitaev is a CSS code that encodes $2$ logical qubits into $n$ physical qubits and achieves a distance of $\sqrt{n/2}$ (a variant of the toric code \cite{BM07} exists with improved parameters $[[n,2,\sqrt{n}]]$). We briefly review the construction of the Kitaev code and its decoding.

\subsection{The toric code}
\label{subsec:toric code}

For $m \in \N$, consider the graph $\bT = (V, E)$ that is the Cayley graph of the additive group $\Z / m\Z \times \Z / m\Z$ with generators $(\pm 1, 0)$ and $(0, \pm 1)$. Specifically, for $m \geq 2$, this is a 4-regular graph that tiles a $2$-dimensional torus by squares. The graph has $|V| = m^{2}$ vertices and $|E| = n = 2m^{2}$ edges. 

The two matrices $\bH_{X}$ and $\bH_{Z}$ are of size $n/2 \times n$ and their columns and rows should be thought of as indexed respectively by the edges and vertices of the graph $\bT$. The matrix $\bH_X$ is simply the vertex-edge incidence matrix of $\bT$ and its rows describe elementary cocycles of the graph. Spelling it out, every vertex $(x, y) \in \Z / m\Z \times \Z / m\Z$ gives rise to a row of $\bH_{X}$ that is a binary vector of weight 4, and whose 1-coordinates are indexed by the edges that connect $(x, y)$ to $(x+1, y),(x, y+1),(x-1, y),(x, y-1)$. The rows of $\bH_{Z}$ correspond to elementary cycles, or faces of the graph, meaning all 4-cycles of the form $(x,y)-(x,y+1)-(x+1,y+1)-(x+1,y)-(x,y)$.

Any elementary cocycle has an even (0 or 2) number of edges in common with a face, which means that every row of $\bH_{X}$ is orthogonal to every row of $\bH_{Z}$. This is exactly the property we need for $\bH_{X}$ and $\bH_{Z}$ to define a quantum CSS code. 

\begin{figure}[htbp]
\centering
\begin{tikzpicture}[scale=0.8]
\draw[step=1] (0,0) grid (5,5);
\draw[line width=0.75mm] (1,2) -- (3,2);
\draw[line width=0.75mm] (2,1) -- (2,3);
\draw[line width=0.75mm] (3,3) -- (4,3);
\draw[line width=0.75mm] (3,3) -- (3,4);
\draw[line width=0.75mm] (4,4) -- (3,4);
\draw[line width=0.75mm] (4,4) -- (4,3);
\end{tikzpicture}
\caption{Cayley graph $\bT$ of the group $\Z / 5\Z \times \Z / 5\Z$ tiling a $2$-dimensional torus. The thick edges represent an elementary cocycle (left) and an elementary cycle (right).}
\label{fig:toric code}
\end{figure}
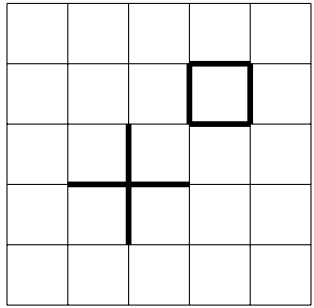 

\textbf{Code dimension.} Since every column of $\bH_{X}$ has weight $2$, the sum of all the rows of $\bH_{X}$ is $0$. The same argument applies to $\bH_{Z}$ and one can check that the dimension of both the row spaces of $\bH_{X}$ and $\bH_Z$ is $|V|-1$. This yields a dimension of $|E|-2|V|+2 = 2$ for the Kitaev code. 

\textbf{Minimum distance.} The $X$-minimum distance $d_X$ is equal to the smallest length of a cycle of $\bT$ that is not a sum of elementary faces. The smallest length of such a homologically non-trivial cycle is seen to be $m$ and is given by a cycle that keeps one coordinate constant. The $Z$-minimum distance $d_Z$ is to equal $d_X$ since elementary faces are elementary cocycles of the dual graph of $\bT$, which is isomorphic to $\bT$. Hence, the minimum distance of the code is $d = m$. 

Summarising, Kitaev's toric code is a quantum error-correcting CSS code with geometrically local parity-check operators of weight $4$ and parameters $[[2m^{2},2,m]] = [[n,2,\sqrt{n/2}]]$. 

\subsection{Decoding the toric code}
\label{subsec:decoding toric code}

The graph $\bT$ and its dual are isomorphic so the decoding problem for $X$-errors is identical as that for \mbox{$Z$-errors}, and we will just focus on one type of error. A $Z$-error pattern is described by a binary vector of length~$n$, $\e \in \F_2^n$. A decoder takes as input the {\em syndrome} of the error $\e$, $\sigma(\e) := \bH_X \e^\intercal$, and outputs a binary vector $\he \in \F_2^n$. The decoder succeeds if $\he = \e + \mathbf{f}$, where $\mathbf{f}$ is in the row space of $\bH_Z$, equivalently, if $\mathbf{f}$ is a homologically trivial cycle, i.e. is a sum of faces of $\bT$. 

Note that the decoder should not require further input, since we are interested in adversarial errors, for which there is no a priori information on the error $\e$ which would come from a channel model for example. 

One should keep in mind that the syndrome vector $\s = \sigma(\e)$ has its coordinates indexed by the vertex set $V$ and can therefore be identified with a subset $S$ of vertices, of even cardinality. The decoder output $\he$ can always be thought of, up to an addition of trivial cycles, as a union of edge-disjoint paths whose endpoints lie in $S$. In \mbox{particular}, what the Minimum Weight Perfect Matching decoder does, is find $\he$ of minimum Hamming weight such that \mbox{$\sigma(\he)=\sigma(\e)$}. Renormalisation decoders use a different approach to output a solution $\he$. We describe their general strategy in the next section and introduce a simple hard-decision and deterministic version of a renormalisation decoder.

\section{A renormalisation decoder}
\label{sec:renormalisation decoder}

The renormalisation idea presented by Duclos-Cianci and Poulin in \cite{DP10[1], DP10[2]} amounts to have a layer of codes together with a simple procedure that passes the error model from one level to the next. One is then able to recursively solve the decoding problem by \emph{renormalising} the error model at each step. The decoders of Duclos-Cianci and Poulin are somewhat loosely defined in the sense that they allow local procedures to adapt to the channel model. To study worst-case behaviour, we need a deterministic version of the decoder that is defined independently of any channel. Such a renormalisation decoder is described bellow.

\subsection{High-level description of the decoder}
\label{subsec:high-level description}

For $k \in \N$, consider the group $V_{k} = \Z/2^{k}\Z \times \Z/2^{k}\Z$ and the graph $\bT_{k} = (V_{k}, E_{k})$ used to construct Kitaev's code. For $i \in [[0,k-1]]$, consider the Cayley graph $\bT_{i} = (V_{i}, E_{i})$, where $V_{i}$ is the subgroup of $\Z/2^{k}\Z \times \Z/2^{k}\Z$ generated by $(\pm 2^{k-i}, 0)$ and $(0, \pm 2^{k-i})$. The vertices $V_i$ of $\bT_i$ can be thought of as a sublattice of $V_{i+1}$, and if we connect vertices of $V_i$ that are at distance $2$ in the graph $\bT_{i+1}$, we obtain the graph $\bT_{i}$. In this process we have {\em renormalised} paths of length~$2$ to paths of length $1$. 

Consider the syndrome vector $\s_{k} = \sigma(\e)$ of an error $\e$, corresponding to a set $S_{k} \subset V_{k}$ containing an even number of vertices of $\bT_{k}$. The decoding problem consists in finding a vector $\he$, identified by a set of edges of $\bT_{k}$, satisfying $\sigma(\he) = \s_{k}$. Suppose first that $S_{k}$ is a subset of $V_{k-1}$. Since an edge of $\bT_{k-1}$ corresponds to a path of length $2$ in $\bT_{k}$, if there exists an edge-pattern in $\bT_{k-1}$ that is consistent with $\s_{k}$, then it induces an edge-pattern in $\bT_{k}$ whose characteristic vector $\he$ satisfies $\sigma(\he) = \s_{k}$. As a result, we can reduce the original decoding problem to a decoding problem in the smaller graph $\bT_{k-1}$. Suppose now that $S_{k}$ is not included in $V_{k-1}$. Then we locally pair up certain vertices of $S_{k}$ and shift others when needed so that the new error syndrome becomes a subset of $V_{k-1}$. 

The edges used in the reduction process are given by a vector $\he_{k}$ satisfying $\s_{k-1} := \sigma(\he_{k}) + \s_{k} \subset V_{k-1}$, where we identify the syndrome $\s_{k-1}$ and its corresponding set of vertices. This reduction process allows us to recursively solve the decoding problem. The complete decoding procedure on the original graph $\bT_{k}$ is given by the following iterative algorithm.

\begin{algorithm}
\caption{Renormalisation decoder}
\label{alg:decoder}
\begin{algorithmic}[0]
\State \textbf{Input:} $\s_{k} = \sigma(\e) \in \F_2^n$, the syndrome of the error $\e$
\State \textbf{Output:} $\he \in \F_2^n$ satisfying $\sigma(\he) = \s_{k}$
\vspace*{0.1cm}
\State \textbf{Initialisation:} set $\he \gets \mathbf{0}$, $\s \gets \s_{k}$, $i \gets k$
\vspace*{0.1cm}
\While{$\s \neq \mathbf{0}$ and $i > 0$}
\State execute the \emph{reduction procedure} on the graph $\bT_{i}$
\State get $\he_{i}$ satisfying $\s_{i-1} := \sigma(\he_{i}) + \s \subset V_{i-1}$
\State set $\he \gets \he + \he_{i}$, $\s \gets \s_{i-1}$, $i \gets i-1$
\EndWhile
\vspace*{0.1cm}
\State \textbf{return} $\he$
\end{algorithmic}
\end{algorithm}

After $k$ steps, the decoder outputs $\he = \sum_{i=1}^{k} \he_{i}$. It remains to check that the output vector $\he$ has the same syndrome has the error vector $\e$, that is to say $\sigma(\he) = \s_{k}$. Observe that for every $j \in [[1,k]]$, \[\s_{j} + \sum_{i=1}^{j} \sigma(\he_i) = \s_{j} + \sigma(\he_j) + \sum_{i=1}^{j-1} \sigma(\he_i) = \s_{j-1} + \sum_{i=1}^{j-1} \sigma(\he_i).\] Hence, we have $\s_{k} + \sigma(\he) = \s_{k} + \sum_{i=1}^{k} \sigma(\he_i) = \s_{0} \subset V_{0}$. The set $V_0$ has only one vertex and since syndrome vectors have even weight, we conclude that $\s_{0} = \mathbf{0}$ and so $\sigma(\he) = \s_{k}$.

\subsection{Precise description of the decoder}
\label{subsec:precise description}

We now describe how to obtain the vectors $\he_{i}$ from the syndrome vectors $\s_{i}$. The original procedure of \cite{DP10[1],DP10[2]} relied on a priori probabilities for coordinates of the original error vector $\e$. Since we are interested in worst-case behaviour, we want a purely deterministic procedure that will work independently of any channel model. We propose the algorithm described below, that is arguably close to the simplest possible from a computational complexity point of view. 

\textbf{Blocks and cells.} In order to precisely define the reduction process of the decoder, we introduce the following notions of blocks and cells. A \emph{cell} of $\bT_{i}$ is defined as an elementary cycle, or face of the graph. A \emph{block} of $\bT_{i}$ is defined as a cell of the graph $\bT_{i-1}$. By construction of the graphs $\bT_{i}$, each block is composed of exactly $4$ cells denoted $A, B, C$ and $D$. For each cell of the tiling, we denote by $l, r, t$ and $b$ the left, right, top and bottom edges of the cell. Moreover, we respectively denote by $\alpha, \beta, \gamma$ and $\delta$ the top-left, top-right, bottom-left and bottom-right vertices of the given cell. These definitions and there relations are shown explicitly in Figure~\ref{fig:precise description}. 

Starting from the graph $\bT_k$, we can iteratively consider the blocks of each subtiling $\bT_i$ and view them as the cells of $\bT_{i-1}$. This allows one to perceive the tower of subtilings that lies at the heart of the renormalisation idea. 

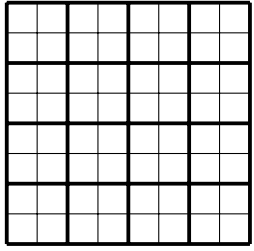
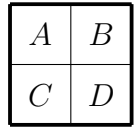
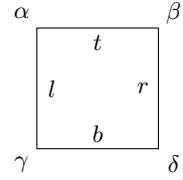
\begin{figure}[htbp]
\centering
\begin{subfigure}[b]{0.4\textwidth}
\centering
\begin{tikzpicture}[scale=0.4]
\draw[step=1] (0,0) grid (8,8);
\draw[step=2,line width=0.5mm] (0,0) grid (8,8);
\end{tikzpicture}
\caption{Graph $\bT_{3}$ tiling a torus by squares. The thick edges represent the blocks of the graph $\bT_{3}$, or equivalently, the cells of the graph $\bT_{2}$.}
\label{subfig:renormalisation}
\end{subfigure}
\hspace*{0.5cm}
\begin{subfigure}[b]{0.24\textwidth}
\centering
\begin{tikzpicture}[scale=0.8]
\draw[step=1] (0,0) grid (2,2);
\draw[step=2,line width=0.5mm] (0,0) grid (2,2);
\node at (0.5,1.5) {{\large $A$}};
\node at (1.5,1.5) {{\large $B$}};
\node at (0.5,0.5) {{\large $C$}};
\node at (1.5,0.5) {{\large $D$}};
\end{tikzpicture}
\vspace*{0.8cm}
\caption{Block of a graph $\bT_i$, subdivided into $4$ cells.}
\label{subfig:block}
\end{subfigure}
\hspace*{0.5cm}
\begin{subfigure}[b]{0.24\textwidth}
\centering
\begin{tikzpicture}[scale=0.8]
\draw[step=2] (0,0) grid (2,2);
\node at (0.25,1) {{\small $l$}};
\node at (1.75,1) {{\small $r$}};
\node at (1,0.25) {{\small $b$}};
\node at (1,1.75) {{\small $t$}};
\node at (-0.25,2.25) {{\small $\alpha$}};
\node at (2.25,2.25) {{\small $\beta$}};
\node at (-0.25,-0.25) {{\small $\gamma$}};
\node at (2.25,-0.25) {{\small $\delta$}};
\end{tikzpicture}
\vspace*{0.4cm}
\caption{Cell of a graph $\bT_i$ with labelled edges and vertices.}
\label{subfig:cell}
\end{subfigure}
\caption{Division of a graph $\bT_i$ into blocks and cells.}
\label{fig:precise description}
\end{figure}

\textbf{Reduction procedure.} To reduce the decoding problem from a syndrome set $S_i \subset V_i$ to a syndrome set $S_{i-1} \subset V_{i-1}$, we proceed in three steps. We start by locally pairing up certain syndrome vertices so as to remove them from the syndrome set. Then, as they may remain vertices in $V_i \setminus V_{i-1}$, we shift those syndrome vertices to move them to the next sublattice. 

In step $1$, we go through all the blocks of $\bT_{i}$ and locally pair up diagonally opposed syndrome vertices in $D$ cells. This is depicted in Figure~\ref{fig:instruction D}.

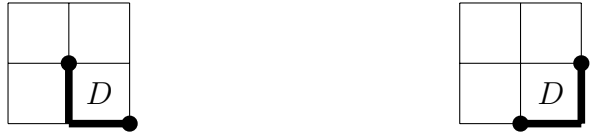
\begin{figure}[htbp]
\centering
\subfloat{
\begin{tikzpicture}[scale=0.8]
\draw[step=1] (0,0) grid (2,2);
\draw[line width=1mm] (1,1) -- (1,0);
\draw[line width=1mm] (1,0) -- (2,0);
\draw[line width=1mm] (1,1) circle (2pt);
\draw[line width=1mm] (2,0) circle (2pt);
\node at (1.5,0.5) {{\large $D$}};
\end{tikzpicture}}
\hfil
\subfloat{
\begin{tikzpicture}[scale=0.8]
\draw[step=1] (0,0) grid (2,2);
\draw[line width=1mm] (2,1) -- (2,0);
\draw[line width=1mm] (1,0) -- (2,0);
\draw[line width=1mm] (2,1) circle (2pt);
\draw[line width=1mm] (1,0) circle (2pt);
\node at (1.5,0.5) {{\large $D$}};
\end{tikzpicture}} 
\caption{Instructions in $D$ cells. Syndrome vertices and their pairing are represented in thicker font. If $\alpha$ and $\delta$ are in $S_i$, then pair them up
via $l$ and $b$. If $\beta$ and $\gamma$ are in $S_i$, then pair them up
via $b$ and $r$.}
\label{fig:instruction D}
\end{figure}

In step $2$, we go through all the blocks of $\bT_{i}$ again and locally pair up certain neighbouring syndrome vertices in $B$ and $C$ cells as shown in Figure~\ref{fig:instruction B,C}.

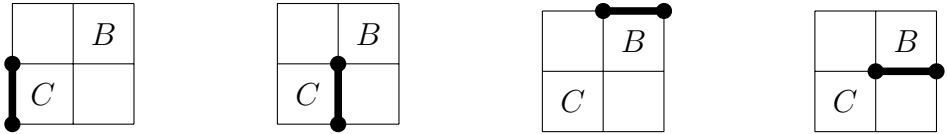
\begin{figure}[htbp]
\centering
\subfloat{
\begin{tikzpicture}[scale=0.8]
\draw[step=1] (0,0) grid (2,2);
\draw[line width=1mm] (0,0) -- (0,1);
\draw[line width=1mm] (0,0) circle (2pt);
\draw[line width=1mm] (0,1) circle (2pt);
\node at (1.5,1.5) {{\large $B$}};
\node at (0.5,0.5) {{\large $C$}};
\end{tikzpicture}}
\hfil
\subfloat{
\begin{tikzpicture}[scale=0.8]
\draw[step=1] (0,0) grid (2,2);
\draw[line width=1mm] (1,0) -- (1,1);
\draw[line width=1mm] (1,0) circle (2pt);
\draw[line width=1mm] (1,1) circle (2pt);
\node at (1.5,1.5) {{\large $B$}};
\node at (0.5,0.5) {{\large $C$}};
\end{tikzpicture}}
\hfil 
\subfloat{
\begin{tikzpicture}[scale=0.8]
\draw[step=1] (0,0) grid (2,2);
\draw[line width=1mm] (1,2) -- (2,2);
\draw[line width=1mm] (1,2) circle (2pt);
\draw[line width=1mm] (2,2) circle (2pt);
\node at (1.5,1.5) {{\large $B$}};
\node at (0.5,0.5) {{\large $C$}};
\end{tikzpicture}}
\hfil 
\subfloat{
\begin{tikzpicture}[scale=0.8]
\draw[step=1] (0,0) grid (2,2);
\draw[line width=1mm] (1,1) -- (2,1);
\draw[line width=1mm] (1,1) circle (2pt);
\draw[line width=1mm] (2,1) circle (2pt);
\node at (1.5,1.5) {{\large $B$}};
\node at (0.5,0.5) {{\large $C$}};
\end{tikzpicture}}
\caption{Instructions in $B$ and $C$ cells. Syndrome vertices and their pairing are represented in thicker font. In the cell $C$, if $\alpha$ and $\gamma$ are in $S_i$, then then pair them up via $l$, if $\beta$ and $\delta$ are in $S_i$, then then pair them up via $r$. In the cell $B$, if $\alpha$ and $\beta$ are in $S_i$, then then pair them up via $t$ and if $\gamma$ and $\delta$ are in $S_i$, then then pair them up via $b$.}
\label{fig:instruction B,C}
\end{figure}

In step $3$, we go through all the blocks of $\bT_{i}$ once more and shift any remaining syndrome vertex in $A$ cells to their top-left corner as shown in Figure~\ref{fig:instruction A}.

\begin{figure}[h!tbp]
\centering
\subfloat{
\begin{tikzpicture}[scale=0.8]
\draw[step=1] (0,0) grid (2,2);
\draw[line width=1mm] (0,2) -- (1,2);
\draw[line width=1mm] (1,2) circle (2pt);
\node at (0.5,1.5) {{\large $A$}};
\end{tikzpicture}}
\hfil
\subfloat{
\begin{tikzpicture}[scale=0.8]
\draw[step=1] (0,0) grid (2,2);
\draw[line width=1mm] (0,2) -- (0,1);
\draw[line width=1mm] (0,1) circle (2pt);
\node at (0.5,1.5) {{\large $A$}};
\end{tikzpicture}}
\hfil 
\subfloat{
\begin{tikzpicture}[scale=0.8]
\draw[step=1] (0,0) grid (2,2);
\draw[line width=1mm] (0,2) -- (0,1);
\draw[line width=1mm] (0,1) -- (1,1);
\draw[line width=1mm] (1,1) circle (2pt);
\node at (0.5,1.5) {{\large $A$}};
\end{tikzpicture}} 
\caption{Instructions in $A$ cells. Syndrome vertices and their travel path are represented in thicker font. If $\beta$ is in $S_i$, then shift it to $\alpha$ via $t$. If $\gamma$ is in $S_i$, then shift it to $\alpha$ via $l$. If $\delta$ is in $S_i$, then shift it to $\alpha$ via $l$ and $b$.}
\label{fig:instruction A}
\end{figure}
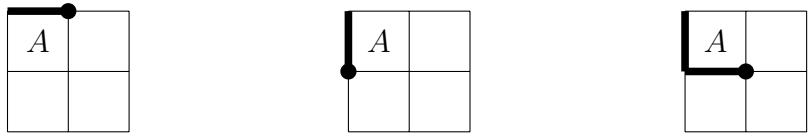

Note that the choice of edges that pair up diagonally opposed vertices in $D$ cells or diagonally shift vertices in $A$ cells is somewhat arbitrary. This doesn't matter however, since other shortest paths are equivalent to the ones we chose. 

All the edges we used in these three steps are given by a vector $\he_{i} \in \F_2^n$. More precisely, starting from the zero vector, the construction of $\he_i$ is done by flipping the bits of $\he_i$ that are indexed by those edges. The reduction procedure is then given by the following iterative algorithm where we go through all the blocks of $\bT_i$ three times.

\begin{algorithm}[h!tbp]
\caption{Reduction procedure}
\label{alg:reduction}
\begin{algorithmic}[0]
\State \textbf{Input:} $\s_{i} \in  \F_2^n$ corresponding to a set of syndrome vertices $S_i \subset V_{i}$
\State \textbf{Output:} $\he_{i} \in \F_2^n$ such that $\s_{i-1} := \sigma(\he_{i}) + \s_{i}$ corresponds to a set of vertices $S_{i-1} \subset V_{i-1}$
\vspace*{0.1cm}
\State \textbf{Initialisation:} set $\he_i \gets \mathbf{0}$, $\s \gets \s_i$ and let $S$ be the vertex set corresponding to $\s$
\vspace*{0.1cm}
\For {every block in $\bT_i$} \Comment Step 1 \hspace*{10cm} \phantom{ }
\State In the $D$ cell:
\State \hspace*{0.5cm} If $\alpha$ and $\delta$ are in $S$, then flip the bits of $\he_i$ indexed by $l$ and $b$
\State \hspace*{0.5cm} If $\beta$ and $\gamma$ are in $S$, then flip the bits of $\he_i$ indexed by $b$ and $r$
\State \hspace*{0.5cm} \textit{Remark: if $\alpha, \beta, \gamma, \delta$ are all in $S$, then we only flip the bits indexed by $l$ and $r$}
\EndFor
\vspace*{0.1cm}
\State Update the syndrome: add to $\s$ the characteristic vectors of all the syndrome vertices that have been paired up in order to remove them from the syndrome set $S$
\vspace*{0.1cm}
\For {every block in $\bT_i$} \Comment Step 2 \hspace*{10cm} \phantom{ }
\State In the cell $C$:
\State \hspace*{0.5cm} If $\alpha$ and $\gamma$ are in $S$, then flip the bit of $\he_i$ indexed by $l$   
\State \hspace*{0.5cm} If $\beta$ and $\delta$ are in $S$, then flip the bit of $\he_i$ indexed by $r$
\State In the cell $B$:
\State \hspace*{0.5cm} If $\alpha$ and $\beta$ are in $S$, then flip the bit of $\he_i$ indexed by $t$
\State \hspace*{0.5cm} If $\gamma$ and $\delta$ are in $S$, then flip the bit of $\he_i$ indexed by $b$
\EndFor
\vspace*{0.1cm}
\State Update the syndrome: add to $\s$ the characteristic vectors of all the syndrome vertices that have been paired up in order to remove them from the syndrome set $S$
\vspace*{0.1cm}
\For {every block in $\bT_i$} \Comment Step 3 \hspace*{10cm} \phantom{ }
\State In the $A$ cell:
\State \hspace*{0.5cm} If $\beta$ is in $S$, then flip the bit of $\he_i$ indexed by $t$ 
\State \hspace*{0.5cm} If $\gamma$ is in $S$, then flip the bit of $\he_i$ indexed by $l$
\State \hspace*{0.5cm} If $\delta$ is in $S$, then flip the bits of $\he_i$ indexed by $l$ and $b$
\EndFor
\vspace*{0.1cm}
\State \textbf{return} $\he_i$
\end{algorithmic}
\end{algorithm}

Note that at the start of each step, the current syndrome $\s$ is equal to $\sigma(\he_{i}) + \s_{i}$. This equality is obvious after the initialisation since $\he_i = \mathbf{0}$ and $\s = \s_i$. After step $1$, certain syndrome vertices have been paired up together and the edges that make up the paths are added to $\he_i$. The syndrome of $\he_i$ at this stage thus corresponds exactly to the set of paired up syndrome vertices. Hence, we have that $\s$, which corresponds to the set of syndrome vertices minus the ones that have been paired up, is equal to $\s_{i} + \sigma(\he_{i})$. The same argument shows that the equality remains true after step $2$. 

\begin{proposition}
\label{prop:decoder}
The reduction procedure is well-defined and reduces the decoding \mbox{problem} from a subtiling to the next one.
\end{proposition}

\begin{proof}
To verify that the reduction process is indeed well-defined, we need to check that there are no conflicting instructions. More precisely, we need to verify that, given a syndrome set at the beginning of a step, we don't pair a syndrome vertex with more than one other vertex or that we don't shift it several times.

In step $1$ the instructions are well-defined since $D$ cells of different blocks share no vertices. Hence, we only pair up a syndrome vertex at most once as shown in Figure~\ref{fig:instruction D}. At step $2$ one has to be a bit more careful since $B$ and $C$ cells share vertices within a block and also between different blocks. Suppose that a shared syndrome vertex $v$ is affected by both instructions in a $B$ and a $C$ cell. This means that there exists two adjacent syndrome vertices to $v$, $v_{B}$ and $v_{C}$, such that $v_{B}$ belongs to $B$ and $v_{C}$ belongs to $C$. Then $v_{B}$ and $v_{C}$ are two diagonally opposite syndrome vertices of a $D$ cell. This is impossible since we removed all such syndrome vertices in step $1$. The multi-step schedule thus ensures that the instructions are done as intended and we pair up a syndrome vertex at most once as shown in Figure~\ref{fig:instruction B,C}. Finally, the instructions are also well-defined in step $3$ since $A$ cells of different blocks share no vertices. Thus, we shift a syndrome vertex only once as shown in Figure~\ref{fig:instruction A}. 

It now remains to verify that the reduction procedure reduces the decoding from a subtiling to the next one. We thus have to check that given the syndrome vector $\s_i$, corresponding to a syndrome set in $V_i$, we obtain an output vector $\he_i$ such that \mbox{$\s_{i-1} := \sigma(\he_{i}) + \s_{i}$} corresponds to a syndrome set in $V_{i-1}$.

After steps $1$ and $2$, certain syndrome vertices have been paired up together and the edges that make up the paths are added to $\he_i$. The syndrome of $\he_i$ at this stage thus corresponds exactly to the set of paired up syndrome vertices. Let us now see what happens in step $3$. Observe that the set of all the vertices we consider in $A$ cells corresponds exactly to $V_i \setminus V_{i-1}$, i.e. to all the vertices that are not corner vertices of a block. In step $3$, we thus shift any remaining syndrome vertex that is not in $V_{i-1}$ to the top-left corner of its corresponding block. The edges that make up the paths that connect the syndrome vertices to their corresponding corner vertices are then added to $\he_i$. Since the syndrome vertices have all been pushed into the vertex set of the next subtiling, the sum \mbox{$\s_{i-1} := \sigma(\he_{i}) + \s_{i}$} indeed corresponds to a vertex subset of $V_{i-1}$.
\end{proof}

Note that in the reduction procedure, we update the current syndrome between each step. One might ask why we don't update it immediately, after having paired up or shifted a syndrome vertex, or after having visited a block. While this doesn't affect the outcome of the algorithm, it is relevant from a computational complexity point of view and allows the procedure to be parallelised.

\begin{proposition}
\label{prop:time-complexity}
The decoder's time-complexity is $O(n \log_2 n)$ and can be parallelised to $O(\log_2 n)$, where $n$ is the length of the toric code. 
\end{proposition}

\begin{proof}
For every $i \in [[1,k]]$, the reduction procedure considers every block of $\bT_{i}$ three times. Within a block, a constant number of operations are executed on the vertices, to test if they belong to the error syndrome, and on the edges, to flip the corresponding bits if needed. The total number of operations in each reduction procedure is therefore linearly bounded by $n$, the size of the graph $\bT_{k}$. Since the recursion ends after $\frac{1}{2} \log_2 n$ iterations, the overall time-complexity is $O(n \log_2 n)$.

The multi-step schedule guarantees that the instructions of the reduction procedure act independently of one another. Moreover, since they are given locally inside blocks and since the current syndrome is updated only between steps of a reduction stage, the instructions can be executed simultaneously. As a result, the decoding scheme can be parallelised and one can achieve a running time proportional to $\log_2 n$, the number of reductions.  
\end{proof}

Finally, we remark that one should not define the decoder in too simple a way. For example, we could question the need for steps $1$ and $2$ in the reduction procedure since step $3$ alone ensures the reduction process by pushing all improperly placed syndrome vertices of a block towards its top-left corner. The reason is that omitting steps $1$ and $2$ would allow a degenerate propagation of error patterns of small weight. We have not found a reduction procedure in less than $3$ steps that avoids conflicting instructions between blocks and avoids wrongly decoded constant-weight error patterns.

For instance, suppose that we do not follow the instructions given in step $2$. Consider the error $\e$, in the graph $\bT_{k}$, consisting only of the top edge $t$ of a $B$ cell of some block. Then the syndrome set $S_{k}$ consists of the end-vertices of that edge. Instead of pairing these two syndrome vertices together immediately, the instructions in step $3$ shift the left end-vertex of the edge to the vertex $\alpha$ of the $A$ cell of that block. Hence, the sum of the error pattern and of the pattern given by the decoder is now composed of the two edges that connect the top vertices of the given block. After renormalisation, we then have a combined pattern of weight $2$ that is the edge of a cell in $\bT_{k-1}$. If this edge happens to be the top edge $t$ of some $B$ cell, then we can repeat the previous argument. We thus can exhibit an error pattern of weight $1$ that, if placed in the correct position, doubles in size at every reduction step. This is problematic since it allows for uncorrectable error patterns of constant weight. Similarly, if we omit step $1$ from the instructions, then one can find an example of an error pattern of weight $2$ that expands and is wrongly decoded.

\section{Analysis over the bit-flip channel}
\label{sec:bit-flip channel}

Before analysing our simple hard-decision renormalisation decoder over the adversarial channel, it is natural to ask how it performs against random errors. We test its performance on the bit-flip channel and numerically determine the threshold probability $p_{0}$ of our decoder using Monte Carlo simulations. 

\begin{figure}[h!t]
\centering
\includegraphics[width=0.75\textwidth]{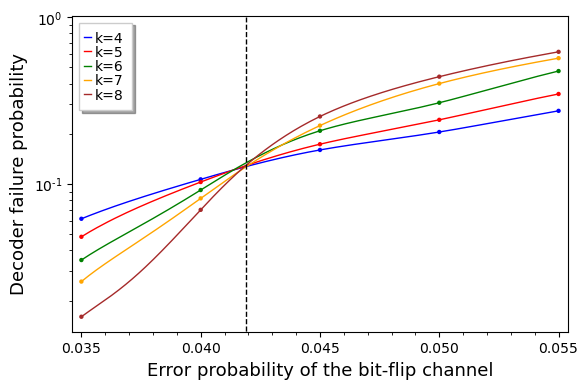}
\caption{Results of the Monte Carlo simulations over the bit-flip channel. The failure probability of our decoder is computed by simulating $5000$ decoding cycles for each $k \in \lbrace 4, 5, 6, 7, 8 \rbrace$ and for each $p \in \lbrace 0.035, 0.04, 0.045, 0.05, 0.055 \rbrace$.}
\label{fig:results}
\end{figure}

The results of the simulations are presented in Figure~\ref{fig:results} and we estimate that the threshold probability is $p_{0} \sim 4.2 \%$. As a comparison, the renormalisation group decoder of Duclos-Cianci and Poulin in its simplest form, without being enhanced by belief propagation techniques, yields a threshold of $p_{0} \sim 7.8\%$ over the depolarisation channel \cite{DP10[1]}, which translates into a threshold probability for the bit-flip channel of $p_{0} \sim 5.2\%$. The performance of the present decoder is therefore reasonably close, and the discrepancy can be explained in part by the fact that our simplified decoder does not use the a priori transition probability of the channel. We conclude that our simple version has captured the essence of a renormalisation decoder.

\section{Analysis over the adversarial channel}
\label{sec:adversarial channel}

We now turn to analysing worst-case errors. More precisely, we are interested in finding the \emph{error-correcting radius} of the decoder, i.e. the largest integer $\omega$ such that any error pattern $\e$ of Hamming weight $\wt(\e) \leq \omega$ is decoded correctly. 
 
For $k \in \N$, consider the tiling $\bT_{k}$ and for $i \in [[0,k-1]]$, the subtilings $\bT_{i}$. Suppose we are given an error vector $\e$ and consider its error syndrome $\s_{k} \subset V_{k}$. Recall that the decoding algorithm defines $k$ intermediate vectors $\he_k,\ldots,\he_1$ which it then adds up to produce the output vector $\he$. Let us define $\e_k=\e$ to be the original error vector, and inductively, for $i=k-1,\ldots, 0$, $\ \e_i := \e_{i+1}+\he_{i+1}$. 

First note that one obtains $\e_{i}$ from $\e_{i+1}$ after one reduction stage since by definition $\e_i = \e_{i+1}+\he_{i+1}$. We say that $\e_{i+1}$ is a \emph{preimage} of $\e_i$. Secondly, remark that an ``intermediate'' error vector $\e_i$ has its syndrome in the vertex set $V_i$ since \[ \sigma(\e_i) = \sigma(\e_{k} + \sum_{j=i+1}^{k} \he_{j}) = \s_k + \sum_{j=i+1}^{k} \sigma(\he_{j}) = \s_i \subset V_i.\] Finally, note that the edges corresponding to $\e_{0} = \e + \he$ form a cycle since at the end of the decoding scheme, both the vectors $\e$ and $\he$ have the same syndrome. Moreover, this cycle is homologically non-trivial if and only if the decoder wrongly decodes the original error vector $\e_k$. 

To analyse the error-correcting radius of the decoder, we will work by starting from a non-trivial cycle $\e_0$ and track how slowly the weights of the errors $\e_i$ can grow when we go back in time to reverse-engineer the decoder's actions (hence the choice of indexation).

\subsection{A wrongly decoded error pattern of weight \texorpdfstring{$\sim d^{1/2}$}{}}
\label{subsec:wrongly decoded error}

We derive an upper bound on the error-correcting radius $\omega$ by finding a small-weight error pattern $\e_{k}$ that is decoded incorrectly, i.e. such that $\e_{0}$ is a homologically non-trivial cycle. 

It is somewhat natural to look for such an error pattern within a minimal cycle of the graph $\bT_{k}$, so we construct it on the first horizontal ``line'' of the graph. If we restrict ourselves to this line, we only consider steps $2$ and $3$ of the reduction process. Moreover, since the error syndrome will be a subset of the vertices of the first line, the decoder only adds edges of this line to the output vector $\he$. 

The idea is to define the vector $\e_1$ as a half a cycle of $\bT_1$ (which is liable to be wrongly decoded by any decoder), then, as we move up the indexes $i$, this half-cycle is expanded into half-cycles of $\bT_i$ but we may regularly introduce ``holes'' in them, creating a fractal-like structure. The construction of $\e_{5}$ starting from a non-trivial cycle $\e_{0}$ is shown in Figure~\ref{fig:wrongly decoded error} and it is easily checked that we indeed obtain $\e_{0}$ when we apply the renormalisation decoder to $\e_{5}$.

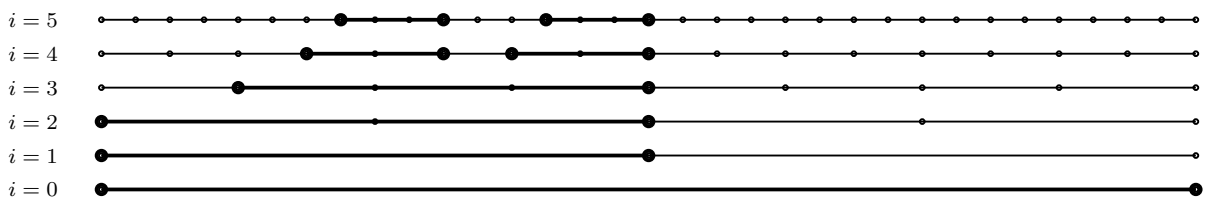
\begin{figure}[htbp]
\centering
\begin{tikzpicture}[scale=0.9]
\foreach \i in {0,...,5}{
	\node at (0,\i /2) {{\footnotesize $i=\i$}};
    \draw[line width=0.25mm] (1,\i /2)  -- (17,\i /2);
}

\foreach \j in {1,17}{
    		\draw[line width=0.25mm] (\j,0) circle (1pt);
}
\foreach \j in {1,9,17}{
    		\draw[line width=0.25mm] (\j,0.5) circle (1pt);
}
\foreach \j in {1,5,...,17}{
    		\draw[line width=0.25mm] (\j,1) circle (1pt);
}
\foreach \j in {1,3,...,17}{
    		\draw[line width=0.25mm] (\j,1.5) circle (1pt);
}
\foreach \j in {1,...,17}{
    		\draw[line width=0.25mm] (\j,2) circle (1pt);
}
\foreach \j in {1,1.5,...,17}{
    		\draw[line width=0.25mm] (\j,2.5) circle (1pt);
}
\draw[line width=0.5mm] (1,0) -- (17,0);
\draw[line width=0.5mm] (1,0.5) -- (9,0.5);
\draw[line width=0.5mm] (1,1) -- (9,1);
\draw[line width=0.5mm] (3,1.5) -- (9,1.5);
\draw[line width=0.5mm] (4,2) -- (6,2);
\draw[line width=0.5mm] (7,2) -- (9,2);
\draw[line width=0.5mm] (4.5,2.5) -- (6,2.5);
\draw[line width=0.5mm] (7.5,2.5) -- (9,2.5);
\draw[line width=0.5mm] (1,0) circle (2pt);
\draw[line width=0.5mm] (17,0) circle (2pt);
\draw[line width=0.5mm] (1,0.5) circle (2pt);
\draw[line width=0.5mm] (9,0.5) circle (2pt);
\draw[line width=0.5mm] (1,1) circle (2pt);
\draw[line width=0.5mm] (9,1) circle (2pt);
\draw[line width=0.5mm] (3,1.5) circle (2pt);
\draw[line width=0.5mm] (9,1.5) circle (2pt);
\draw[line width=0.5mm] (4,2) circle (2pt);
\draw[line width=0.5mm] (6,2) circle (2pt);
\draw[line width=0.5mm] (7,2) circle (2pt);
\draw[line width=0.5mm] (9,2) circle (2pt);
\draw[line width=0.5mm] (4.5,2.5) circle (2pt);
\draw[line width=0.5mm] (6,2.5) circle (2pt);
\draw[line width=0.5mm] (7.5,2.5) circle (2pt);
\draw[line width=0.5mm] (9,2.5) circle (2pt);
\end{tikzpicture}
\caption{Construction of the wrongly decoded error pattern. For $i \in [[0,5]]$, the error $\e_{i}$ is shown in the first horizontal line of the graph $\bT_i$.}
\label{fig:wrongly decoded error}
\end{figure} 

Precisely, the first steps of the construction are as follows. Start from the non-trivial cycle $\e_{0}$ that is located on the first horizontal line of the graph $\bT_{k}$. This corresponds to the first horizontal edge of the graph $\bT_{0}$. The first horizontal line of $\bT_{1}$ is composed of $2$ edges and we define $\e_{1}$ by the $1^{st}$ edge. Note that if we decode $\e_{1}$ in the graph $\bT_{1}$, then we obtain $\e_{0}$ since at step $2$, the decoder pairs the two endpoints of $\e_{1}$ together using the $2^{nd}$ edge. The first horizontal line of $\bT_{2}$ is composed of $4$ edges and we define $\e_{2}$ by the first two edges. Observe that if we decode $\e_{2}$ in the graph $\bT_{2}$, then we trivially obtain $\e_{1}$ by renormalisation. 

Let us now see how to construct the errors $\e_i$ for $i \geq 5$. Note that $\e_{4}$ is composed of $2$ paths, each of weight $2$. By repeating the same method that is used to obtain $\e_{3}$ from $\e_{2}$ to each path of $\e_{4}$, we obtain the error $\e_{5}$. This results in a pattern on the first horizontal line of $\bT_{5}$ with $2$ paths, each of weight $3$. Similarly, we obtain $\e_{6}$ by repeating the same procedure that is used to obtain $\e_{4}$ from $\e_{3}$ to each path of $\e_{5}$. This results in a pattern on the first horizontal line of $\bT_{6}$ with $4$ paths, each of weight $2$. We thus can iteratively construct the following error patterns $\e_{i}$ by transforming each path of weight $2$ to a path of weight $3$, and each path of weight $3$ to $2$ paths of weight $2$.  

Computing the weights of the constructed error patterns $\e_{k}$ for $k \in \N$ leads to the following statement.

\begin{proposition}
\label{prop:upper bound}
Consider $(u_{k})_{k \in \N}$ the sequence defined by $u_{k} = k$ for $k \in \lbrace 1,2,3 \rbrace$ and $u_{k}=2u_{k-2}$ for $k \geq 4$. For $k \in \N$, let $\omega_k$ be the largest integer such that any error pattern $\e_{k}$ in the graph $\bT_{k}$ of weight $\wt(\e_{k}) \leq \omega_k$ is decoded correctly. Then $\forall \, k \in \N, \ \omega_k \leq u_{k}-1$.  
\end{proposition} 

\begin{proof}
We compute the weight of the errors $\e_{k}$ constructed above and show by induction that $\forall \, k \in \N$, $\wt(\e_{k}) = u_{k}$. Since the errors $\e_{k}$ are wrongly decoded, it will follow that $\forall \, k \in \N$, $\omega_k < u_{k}$. 

The equality $\wt(\e_{k}) = u_{k}$ is immediately verified for $k \in \lbrace 1,2,3 \rbrace$. Now consider the constructed error $\e_k$ for $k \geq 4$. Observe that either $\e_{k-2}$ is composed only of paths of weight $2$ or $\e_{k-2}$ is composed only of paths of weight $3$. In the first case, we have that by construction every path is transformed into a path of weight $3$, which is in turn transformed into $2$ paths of weight $2$. Hence, $\e_{k}$ has exactly twice as much paths of weight $2$ as $\e_{k-2}$. In the second case, we observe similarly that $\e_{k}$ has exactly twice as many paths of weight $3$ as $\e_{k-2}$. In both cases $\wt(\e_{k}) = 2\wt(\e_{k-2})$ and so we can conclude that $\wt(\e_{k}) = 2\wt(\e_{k-2}) = 2u_{k-2} = u_{k}$.
\end{proof}

We observe that the sequence $(u_{k})_{k \in \N}$ scales as $(\sqrt{2^k})_{k \in \N}$ since for every $k  \geq4, \; u_{k} = 2u_{k-2}$. Given that $\sqrt{2^k} = d^{1/2}$, the error-correcting radius $\omega$ is bounded from above by a quantity that scales as $d^{1/2}$. 

\subsection{\texorpdfstring{$1$}{}-dimensional study of the error-correcting radius}
\label{subsec:1-dimensional study}

Finding a lower bound for the error-correcting radius $\omega$ is not as straightforward as it might initially seem. Let us therefore start by considering the $1$-dimensional case where we only allow errors to be placed on a single line of the toric graphs. 

In this study, cells are reduced to single edges and blocks to two neighbouring edges. Moreover, the renormalisation decoder is simplified and we only consider two steps in each reduction stage. The precisely reduction procedure in the $1$-dimensional case is as follows. We first go through all the blocks of the graph and in each block, if the end vertices of the right edge of the block are in the error syndrome, then we pair them up with this right edge (Figure~\ref{fig:instruction B,C}). We then go through all the blocks of the graph again and in each block, if the middle vertex of the block is in the error syndrome, then we link it to the left vertex of the block with the left edge of that block (Figure~\ref{fig:instruction A}). 

If we restrict ourselves to this $1$-dimensional case, then one can determine the exact value of the error-correcting radius $\omega$. 

\begin{proposition}
\label{prop:lower bound 1D}
Consider $(u_{k})_{k \in \N}$ the sequence defined by $u_{k} = k$ for $k \in \lbrace 1,2,3 \rbrace$ and $u_{k}=2u_{k-2}$ for $k \geq 4$. For $k \in \N$, let $\omega_k$ be the largest integer such that any \mbox{1-dimensional} error pattern $\e_{k}$ in the graph $\bT_{k}$ of weight $\wt(\e_{k}) \leq \omega_k$ is decoded correctly. Then $\forall k \in \N, \ \omega_k = u_{k}-1$. 
\end{proposition}

By Proposition~\ref{prop:upper bound}, $\forall \, k \in \N, \ \omega_k \leq u_{k}-1$ so it remains to show that we have \mbox{$\forall \, k \in \N$}, $\omega_k \geq u_{k}-1$. To achieve this lower bound on the error-correcting radius, recall the notion of \emph{preimage} of an error pattern $\e_{i}$ on the graph $\bT_i$. A preimage of $\e_{i}$ is an error pattern $\e_{i+1}$ on the graph $\bT_{i+1}$ such that one obtains $\e_{i}$ from $\e_{i+1}$ after one reduction stage of the decoding process. The main idea is then to show that the weight of any $1$-dimensional second preimage of a given $1$-dimensional error is twice as large as the weight of the error itself. This is proved in the following three lemmas.

\begin{lemma}
\label{lem:growth 1D, 1}
Let $i \in [[0,k-1]]$ and let $\e_{i}$ be a 1-dimensional error pattern on the graph~$\bT_{i}$. Let $P_{i}$ denote the number of paths of $\e_i$ and let $\e_{i+1}$ be a 1-dimensional preimage of $\e_{i}$. Then $\wt(\e_{i+1}) \geq \wt(\e_{i}) + P_{i}$. 
\end{lemma}

\begin{proof}
Observe that in the $1$-dimensional case, during the reduction process, the decoder modifies the edges of a block of $\bT_{i+1}$ only if the central vertex of that block is in the error syndrome. This happens if exactly one of the two edges of that block is an error edge, i.e. an edge of $\e_{i+1}$. It follows that if no error edge is present in a given block of $\bT_{i+1}$, then after one reduction stage the corresponding edge of $\bT_{i}$ is not an error edge, i.e. not an edge of $\e_i$. Since $\e_{i+1}$ is a preimage of $\e_{i}$, we thus need at least one error edge of $\e_{i+1}$ in each block defined by an error edge of $\e_{i}$. Hence, we have that $\wt(\e_{i+1}) \geq \wt(\e_{i})$. 

Note that in the $1$-dimensional study, an error pattern $\e_i$ is either a non-trivial cycle or a disjoint union of paths. If $\e_i$ is a non-trivial cycle, then $P_i = 0$ and so the desired inequality is satisfied since we just showed that $\wt(\e_{i+1}) \geq \wt(\e_{i})$. Let us now suppose that $\e_i$ is a disjoint union of paths. If $\e_{i}$ consists of more than one path, then two distinct paths are separated by at least two edges in $\bT_{i+1}$ and an error edge of $\e_{i+1}$ located in between the two paths contributes to at most one of them. It follows that we can partition the error edges of $\e_{i+1}$ according to which of the paths of $\e_{i}$ they contribute to. We may thus without loss of generality assume that $\e_{i}$ consists of a single path and prove that $\wt(\e_{i+1}) \geq \wt(\e_{i}) + 1$.

Note that if we put exactly one error edge of $\e_{i+1}$ in each block defined by an error edge of $\e_{i}$ and no error edges elsewhere, then $\e_{i+1}$ is not a preimage of $\e_{i}$. Indeed, consider the error edge of $\e_{i}$ that has no error edge to its right. Note this edge exists and is unique since $\e_{i}$ is not a cycle by assumption. Suppose that the corresponding block of the graph $\bT_{i+1}$ only contains one error edge of $\e_{i+1}$. If this is the right edge of the block, then it is removed since the edge is isolated, i.e. there is no error edge to its left or to its right. There is no error edge to its left because we supposed that there is only error edge in the block and there is no error edge to its right because we supposed that there are no error edges outside the blocks defined by the error edges of $\e_{i}$. If the error edge of $\e_{i+1}$ is the left edge of the block, then it is also removed since its right vertex is in the error syndrome and has not been paired up with an other vertex. This is because we supposed that it is the only error edge in the block and that there are no error edges outside the blocks defined by the error edges of $\e_{i}$. Hence, we cannot put only one error edge of $\e_{i+1}$ in each block defined by an error edge of $\e_{i}$ and no error edges elsewhere. This means that we need at least one error edge more and so $\wt(\e_{i+1}) \geq \wt(\e_{i}) + 1$. 
\end{proof}

\begin{lemma}
\label{lem:growth 1D, 2}
Let $i \in [[0,k-1]]$ and let $\e_{i}$ be a 1-dimensional error pattern on the graph~$\bT_{i}$. Let $\e_{i+1}$ be a 1-dimensional preimage of $\e_{i}$ and let $P_{i+1}$ denote the number of its paths. Then $\wt(\e_{i+1}) + P_{i+1} \geq 2\wt(\e_{i})$. 
\end{lemma}

\begin{proof}
Note that an error edge of $\e_{i+1}$ that is not in a block defined by an error edge of $\e_{i}$ disappears after the reduction process. If such an error edge doesn't contribute to the error $\e_i$, then it can simply be removed to decrease both $\wt(\e_{i+1})$ and $P_{i+1}$. Similarly, if such an error edge does contribute to $\e_i$, then it can be moved inside a block defined by an error edge of $\e_{i}$ to decrease $P_{i+1}$. We may thus assume that every error edge of $\e_{i+1}$ is in a block defined by an error edge of $\e_{i}$. 

To prove the desired inequality, let us construct $\e_{i+1}$ from $\e_i$ and try to minimise $\wt(\e_{i+1}) + P_{i+1}$. We start by supposing that both edges of each block defined by an error edge of $\e_{i}$ are error edges of $\e_{i+1}$.  Then $\wt(\e_{i+1}) + P_{i+1} \geq \wt(\e_{i+1}) = 2\wt(\e_{i})$. We now consider the two following cases where either the error pattern $\e_i$ is a non-trivial cycle or a disjoint union of paths, and see what happens if we start removing error edges from~$\e_{i+1}$.  

Case $1$: $\e_i$ is a non-trivial cycle. Since $\e_i$ is a non-trivial cycle, $\e_{i+1}$ is also a non-trivial cycle at the start of our construction. We thus have $\wt(\e_{i+1}) + P_{i+1} = \wt(\e_{i+1}) = 2\wt(\e_{i})$. Let us now see what happens if we start removing error edges from $\e_{i+1}$. First, remove any error edge of $\e_{i+1}$. Then the weight of $\e_{i+1}$ is reduced by one but the number of paths $P_{i+1}$ is increased by one since $\e_{i+1}$ is now a path. Hence, the sum $\wt(\e_{i+1}) + P_{i+1}$ is unchanged. Note that we cannot remove the rightmost or leftmost error edge of the new error $\e_{i+1}$, i.e. the error edges neighbouring the error edge that was just removed. Indeed, if the removed edge was the left edge of a block, then we cannot remove the edge to its right since there needs to be at least one error edge in the given block. We can also not remove the edge to its left since $\e_{i+1}$ would no longer be a preimage of $\e_{i}$. Similarly, if the removed edge was the right edge of a block, then we cannot remove its neighbouring edges for the same reasons. Consider now removing an error edge of $\e_{i+1}$ that is not an outer edge, i.e. an error edge that has an error edge to its left and to its right. If we remove one of these edges, then the weight of $\e_{i+1}$ is reduced by one but the number of paths $P_{i+1}$ is increased by one. Hence, the sum $\wt(\e_{i+1}) + P_{i+1}$ is again unchanged. We thus have $\wt(\e_{i+1}) + P_{i+1} = 2\wt(\e_{i})$. 

Case $2$: $\e_i$ is a disjoint union of paths. As we assumed that every error edge of $\e_{i+1}$ is in a block defined by an error edge of $\e_{i}$, we can partition the error edges of $\e_{i+1}$ according to which of the paths of $\e_{i}$ they contribute to. We may thus without loss of generality assume that $\e_{i}$ consists of a single path. Then $\e_{i+1}$ consists also of a single path at the start of our construction and $\wt(\e_{i+1}) + P_{i+1} = 2\wt(\e_{i})+1$. Let us now see what happens if we start removing error edges from $\e_{i+1}$. First, consider the rightmost error edge of $\e_{i+1}$, i.e. the one that has no error edge to its right. Note that this error edge exists and is unique since $\e_{i+1}$ is not a cycle by assumption. We cannot remove this error edge since otherwise $\e_{i+1}$ is no longer a preimage of $\e_{i}$. Consider now the leftmost error edge of $\e_{i+1}$, i.e. the one that has no error edge to its left. If we remove this error edge, then the weight of $\e_{i+1}$ is reduced by one and the number of paths $P_{i+1}$ is unchanged. Hence, we have that $\wt(\e_{i+1}) + P_{i+1} = 2\wt(\e_{i})$. We cannot remove the leftmost error edge of our new $\e_{i+1}$ since recall that we need at least one error edge of $\e_{i+1}$ in each block defined by an error edge of $\e_i$. Finally, consider an error edge of $\e_{i+1}$ that is not an outer edge, i.e. an edge that has an error edge to its left and to its right. If we remove one of these edges, then the weight of $\e_{i+1}$ is reduced by one but the number of paths $P_{i+1}$ is increased by one. Hence, the sum $\wt(\e_{i+1}) + P_{i+1}$ is unchanged. We thus have $\wt(\e_{i+1}) + P_{i+1} \geq 2\wt(\e_{i})$. 
\end{proof}

\begin{lemma}
\label{lem:growth 1D, 3}
Let $i \in [[0,k-2]]$ and let $\e_{i}$ be a 1-dimensional error pattern on the graph~$\bT_{i}$. Let $\e_{i+2}$ be a 1-dimensional second preimage of the error $\e_{i}$. Then~$\wt(\e_{i+2}) \geq 2 \wt(\e_{i})$. 
\end{lemma}

\begin{proof}
Consider the image of the error $\e_{i+2}$ after one reduction stage. This image is a $1$-dimensional error pattern which we denote by $\e_{i+1}$. Then $\e_{i+2}$ is clearly a preimage of $\e_{i+1}$ so by Lemma~\ref{lem:growth 1D, 1}, $\wt(\e_{i+2}) \geq \wt(\e_{i+1}) + P_{i+1}$. Now $\e_{i+1}$ is also a preimage of $\e_{i}$ so by Lemma~\ref{lem:growth 1D, 2}, $\wt(\e_{i+1}) + P_{i+1} \geq 2\wt(\e_{i})$. Combining the two inequalities, we obtain $\wt(\e_{i+2}) \geq 2 \wt(\e_{i})$. 
\end{proof}

We now are able to prove Proposition~\ref{prop:lower bound 1D} that states that $\forall k \in \N, \ \omega_k = u_{k}-1$, where $(u_{k})_{k \in \N}$ is the sequence defined by $u_{k} = k$ for $k \in \lbrace 1,2,3 \rbrace$ and $u_{k}=2u_{k-2}$ for $k \geq 4$. 
 
\begin{proof}[Proof of Proposition~\ref{prop:lower bound 1D}]
By Proposition~\ref{prop:upper bound}, $\forall \, k \in \N, \ \omega_k \leq u_{k}-1$ so it remains to show that we have $\forall \, k \in \N,\ \omega_k \geq u_{k}-1$. We thus~need to show that if a $1$-dimensional error pattern $\e_{k}$ in the graph $\bT_{k}$ has weight $\wt(\e_{k}) < u_{k}$, then $\e_{k}$ is decoded correctly. Equivalently, consider a $1$-dimensional error pattern $\e_{k}$ in the graph $\bT_{k}$ that is decoded incorrectly and let us show by strong induction that $\wt(\e_{k}) \geq u_{k}$. 

By exhausting all the possible $1$-dimensional patterns, one sees that for \mbox{$k \in \lbrace 1,2,3 \rbrace$} the minimal weight of an error pattern in $\bT_{k}$ that is decoded incorrectly is equal to $u_{k}$. Hence, for every $k \in \lbrace 1,2,3 \rbrace$, $\wt(\e_{k}) \geq u_{k}$. For $k \geq 4$, consider a wrongly decoded $1$-dimensional error $\e_k$ and the $1$-dimensional pattern $\e_{k-2}$ obtained after two reduction stages. Since $\e_{k-2}$ is also a wrongly decoded $1$-dimensional error pattern,  $\wt(\e_{k-2}) \geq u_{k-2}$. By Lemma~\ref{lem:growth 1D, 3}, we thus have \mbox{$\wt(\e_{k}) \geq 2 \wt(\e_{k-2}) \geq 2u_{k-2}= u_{k}$}.
\end{proof}

\subsection{Lower bound on the error-correcting radius}
\label{subsec:lower bound}

In the $2$-dimensional case, we are not able to determine the exact value of the error-correcting radius but we prove the following lower bound which scales as $\frac{5}{6} d^{\,\log_2 \! \frac{6}{5}} \sim \frac{5}{6} d^{\, 0.263}$.

\begin{theorem}
\label{thm:lower bound}
Consider $(v_{k})_{k \in \N}$ the sequence defined by $v_k=\left( \frac{6}{5} \right)^{k-1}$. For $k \in \N$, let $\omega_k$ be the largest integer such that any error  $\e_{k}$ of weight $\wt(\e_{k}) \leq \omega_k$ is decoded correctly. Then~$\forall k \in \N, \ \omega_k \geq v_{k}-1$. 
\end{theorem}

In the $1$-dimensional case, we showed in Section~\ref{subsec:1-dimensional study} that the weight of any second preimage of an error is twice as large as the weight of the error itself. This technique doesn't extend directly to the $2$-dimensional case because there exist error patterns whose preimages have smaller weight. To avoid this issue, we introduce a notion of {\em reduced weight} for the intermediate errors $\e_i$. 

\textbf{Reduced weight.} To define the reduced weight of an error pattern, we first partition the original error $\e_{k}$ into a union of edge-disjoint paths, whose endpoints are syndrome vertices, and into a union of edge-disjoint cycles. There may be several ways to partition the error $\e_{k}$ but we just choose any one of them. In the first reduction stage, the decoder adds the edges given by $\he_{k}$ to the current error. This modifies the paths and cycles of the combined error $\e_{k-1}$. For example, certain paths may be joined together, some may have their endpoints shifted, and others can disappear or, equivalently, form trivial cycles. Regardless, the added edges given by $\he_{k}$ induce a partition of $\e_{k-1}$ into edge-disjoint paths and cycles. The decoder therefore inductively defines a sequence of partitions of the intermediate errors $\e_i$. We now can define the reduced weight of an error pattern $\e_i$. For any path $\mathbf{p}$ in the partition of $\e_i$ into paths, the reduced weight $\wtr(\mathbf{p})$ of $\mathbf{p}$ is the weight of a shortest path connecting the endpoints of~$\mathbf{p}$. For any cycle $\mathbf{c}$ in the partition of $\e_i$ into cycles, the reduced weight $\wtr(\mathbf{c})$ of $\mathbf{c}$ is the weight of a shortest equivalent cycle. Hence, the reduced weight of a trivial cycle is $0$ and the reduced weight of a non-trivial cycle is $d = 2^{i}$, the girth of the graph $\bT_{i}$. The reduced weight $\wtr(\e_i)$ of $\e_{i}$ is then defined as the sum of the reduced weights of all the paths and cycles composing~$\e_{i}$. 

This notion is still not enough, for it is difficult to control the growth of the reduced weight $\wtr(\e_i)$ for increasing indexes. But whenever the increase from $\wtr(\e_i)$ to $\wtr(\e_{i+1})$ is too small, it is because the number of paths that make up $\e_{i+1}$ has increased significantly. What we therefore do is track the growth of the combined quantity $\wtr(\e_i) + P_i$ where $P_i$ denotes the number of paths in the partition of $\e_i$. The core of the proof of Theorem~\ref{thm:lower bound} lies then in the following lemma.

\begin{lemma}
\label{lem:growth 2D}
For $i \in [[0,k-1]]$, let $\e_{i}$ be an error pattern on the graph $\bT_{i}$ and denote by $P_{i}$ the number of its paths. Let $\e_{i+1}$ be a preimage of $\e_{i}$ and let $P_{i+1}$ denote the number of its paths. Then \mbox{$\wtr(\e_{i+1}) + P_{i+1} \geq \frac{6}{5} (\wtr(\e_{i}) + P_{i})$}.
\end{lemma}

\begin{proof}
Let us choose a partition of $\e_{i+1}$ into a union of edge-disjoint paths and into a union of edge-disjoint cycles and consider the induced partition of $\e_{i}$. Note that an error edge of $\e_{i+1}$ may contribute to only one path or cycle of $\e_{i}$ so we can partition the error edges of $\e_{i+1}$ according to which of the paths or cycles of $\e_{i}$ they contribute to. We may thus without loss of generality assume that $\e_{i}$ either consists of a single path or of a single cycle. 

Case $1$: $\e_i$ is a cycle. If $\e_{i}$ is a homologically trivial cycle, then $\wtr(\e_i) = P_i = 0$ and so the inequality $\wtr(\e_{i+1}) + P_{i+1} \geq \frac{6}{5} (\wtr(\e_{i}) + P_{i})$ is clearly satisfied. Let us now suppose that $\e_{i}$ is a non-trivial cycle. \mbox{We want to prove} that $\wtr(\e_{i+1})+ P_{i+1} \geq \frac{6}{5} \wtr(\e_{i})$. Note that $\wtr(\e_{i})$ is the weight of a minimal non-trivial cycle on the graph $\bT_{i}$ equivalent to $\e_i$. Hence, the inequality is independent of the cycle. We may thus think of $\e_{i}$ as any non-trivial cycle that loops the torus along one of its two principal directions. To prove the desired inequality, let us construct the error pattern $\e_{i+1}$ in such a way that $\wtr(\e_{i+1}) + P_{i+1}$ is minimal. To do this, we start from any minimal non-trivial cycle $\mathbf{c}$ and construct $\e_{i+1}$ by reversing the decoding process. Hence, we start with $\wtr(\e_{i+1}) = 2\wt(\mathbf{c})$ and $P_{i+1}=0$. The situation is now similar to the $1$-dimensional case since the decoder may only remove single-edges of $\mathbf{c}$ (step $2$ and $3$ of the decoder) or move its endpoints by one edge (step $3$ of the decoder). The same arguments as the ones found in the proof of Lemma~\ref{lem:growth 1D, 2} then show that $\wtr(\e_{i+1}) + P_{i+1} \geq 2\wtr(\e_{i})$ and so desired inequality is satisfied. 

Case $2$: $\e_{i}$ is a single path joining its two endpoints, say $a$ and $b$. Let us prove that the inequality $\wtr(\e_{i+1})+ P_{i+1} \geq \frac{6}{5}(\wtr(\e_{i}) + 1)$ is satisfied. Note that $\wtr(\e_{i})$ is just the graph distance between the vertices $a$ and $b$ in the graph $\bT_{i}$, denoted $\mathsf{d_{T}}(a,b)$. Hence, the inequality is independent of the actual path, as long as it joins $a$ to $b$. We may thus think of $\e_{i}$ as any path from $a$ to $b$. To prove the desired inequality, let us construct the error pattern $\e_{i+1}$ in such a way that $\wtr(\e_{i+1}) + P_{i+1}$ is minimal. To do this, we start from any minimal path $\mathbf{p}$ joining $a$ to $b$ and we construct $\e_{i+1}$ by reversing the decoding process. Hence, we start with $\wtr(\e_{i+1}) = 2\mathsf{d_{T}}(a,b)$ and $P_{i+1}=1$. The decoder may remove double-edges of $\mathbf{p}$ (steps $1$ and $3$ of the decoder), remove single-edges of $\mathbf{p}$ (step $2$ and $3$ of the decoder) and move its endpoints by one or two edges (step $3$ of the decoder). If we remove a single-edge, then $\wtr(\e_{i+1})$ decreases by $1$ but $P_{i+1}$ increases by~$1$. Hence, removing a single-edge doesn't affect $\wtr(\e_{i+1})+P_{i+1}$ and so we don't use this to construct $\e_{i+1}$. If we remove a double-edge, then $\wtr(\e_{i+1})$ decreases by $2$ but $P_{i+1}$ increases by only $1$. Hence, removing a double-edge decreases the sum $\wtr(\e_{i+1})+P_{i+1}$ by $1$. As a result, we try to remove as many double-edges of $\mathbf{p}$ as we can. Let $\delta$ denote the number of double-edges we remove. Note that if we remove a double-edge from $\mathbf{p}$, we should at least keep in $\e_{i+1}$ the two edges it is connected to, otherwise we would have removed at least three connected edges. Moreover, it must be that between two removed double-edges, we should keep at least two edges in $\e_{i+1}$. Indeed, consider a section of the path $\mathbf{p}$ where we only keep one edge between two removed double-edges. Two examples of this situation are shown in Figure~\ref{fig:undesirable path}. It is easily checked that these double-edges cannot be parallel (where we identify a double-edge with the pair of its endpoints). This contradicts the minimality of the path $\mathbf{p}$ (we assumed that $\mathbf{p}$ is a minimal path joining $a$ to $b$). We could of course construct $\e_{i+1}$ from a locally modified version of $\mathbf{p}$ to account for this but one can verify that this doesn't lessen $\wtr(\e_{i+1})+P_{i+1}$.

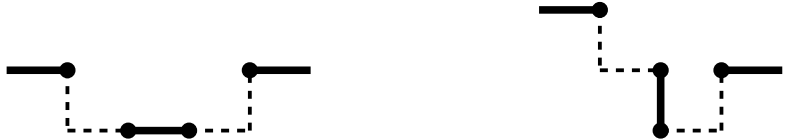
\begin{figure}[htbp]
\centering
\subfloat{
\begin{tikzpicture}[scale=0.8]
\draw[line width=1mm] (0,1) -- (1,1);
\draw[line width=0.5mm, dashed] (1,1) -- (1,0);
\draw[line width=0.5mm, dashed] (1,0) -- (2,0);
\draw[line width=1mm] (2,0) -- (3,0);
\draw[line width=0.5mm, dashed] (3,0) -- (4,0);
\draw[line width=0.5mm, dashed] (4,0) -- (4,1);
\draw[line width=1mm] (4,1) -- (5,1);
\draw[line width=1mm] (1,1) circle (2pt);
\draw[line width=1mm] (2,0) circle (2pt);
\draw[line width=1mm] (3,0) circle (2pt);
\draw[line width=1mm] (4,1) circle (2pt);
\end{tikzpicture}}
\hfil
\subfloat{
\begin{tikzpicture}[scale=0.8]
\draw[line width=1mm] (0,2) -- (1,2);
\draw[line width=0.5mm, dashed] (1,2) -- (1,1);
\draw[line width=0.5mm, dashed] (1,1) -- (2,1);
\draw[line width=1mm] (2,1) -- (2,0);
\draw[line width=0.5mm, dashed] (2,0) -- (3,0);
\draw[line width=0.5mm, dashed] (3,0) -- (3,1);
\draw[line width=1mm] (3,1) -- (4,1);
\draw[line width=1mm] (1,2) circle (2pt);
\draw[line width=1mm] (2,1) circle (2pt);
\draw[line width=1mm] (2,0) circle (2pt);
\draw[line width=1mm] (3,1) circle (2pt);
\end{tikzpicture}}
\caption{Sections of paths where we only keep one edge between two removed double-edges (dashed edges). The removed double-edges are non non-parallel resulting in non-minimal paths.}
\label{fig:undesirable path}
\end{figure}

If we move the endpoints of $\mathbf{p}$, then we can change the distance $2\mathsf{d_{T}}(a,b)$ by at most $2$ since the instructions in $A$ cells move the endpoints in the same direction. We can use this to reduce $\wtr(\e_{i+1})$ but moving the endpoints may also affect the number $\delta$ of double-edges we can remove. Let $\mu \in \lbrace -2,-1,0,1,2 \rbrace$ denote by how much we have modified $2\mathsf{d_{T}}(a,b)$. For each value of $\mu$ we show that $\wtr(\e_{i+1})+P_{i+1} = 2\mathsf{d_{T}}(a,b)+ \mu + 1- \delta \geq \frac{6}{5}(\wtr(\e_{i}) + 1)$. 

Case $\mu = 2$: we have moved the endpoints of $\mathbf{p}$ further apart such that $\wtr(\e_{i+1})$ is increased by $2$. Since we keep at least two edges in $\e_{i+1}$ between two removed double-edges, $4\delta \leq 2\mathsf{d_{T}}(a,b)+2$. Hence, \[ \wtr(\e_{i+1})+P_{i+1} = 2\mathsf{d_{T}}(a,b)+2+1-\delta \geq \frac{3}{2}\mathsf{d_{T}}(a,b) +\frac{5}{2} = \frac{3}{2}\wtr(\e_{i}) +\frac{5}{2} \geq \frac{6}{5}(\wtr(\e_{i}) + 1).\] 

Case $\mu = 1$: we have moved the endpoints of $\mathbf{p}$ further apart such that $\wtr(\e_{i+1})$ is increased by $1$. The same argument as above shows that $4\delta \leq 2\mathsf{d_{T}}(a,b)+1$ and so \[ \wtr(\e_{i+1})+P_{i+1} = 2\mathsf{d_{T}}(a,b)+1+1-\delta \geq \frac{3}{2}\mathsf{d_{T}}(a,b) +\frac{7}{4} = \frac{3}{2}\wtr(\e_{i}) +\frac{7}{4} \geq \frac{6}{5}(\wtr(\e_{i}) + 1).\] 

Case $\mu = 0$: the reduced weight $\wtr(\e_{i+1})$ is unchanged. We have that $4\delta \leq 2\mathsf{d_{T}}(a,b)$ and so \[ \wtr(\e_{i+1})+P_{i+1} = 2\mathsf{d_{T}}(a,b)+1-\delta \geq \frac{3}{2}\mathsf{d_{T}}(a,b) + 1 \geq \frac{5}{4}(\wtr(\e_{i}) + 1) \geq \frac{6}{5}(\wtr(\e_{i}) + 1).\] Indeed, the function, $\wtr(\e_{i}) \mapsto (\frac{3}{2}\wtr(\e_{i}) +1)(\wtr(\e_{i}) +1)^{-1}$ is minimal for $\wtr(\e_{i})=1$ and equals $\frac{5}{4}$.

Case $\mu = -1$: we have moved the endpoints of $\mathbf{p}$ closer together such that $\wtr(\e_{i+1})$ is decreased by~$1$. Note that this requires that $\wtr(\e_{i})$ is at least $2$. To achieve the desired bound, one has to be more careful when deriving an upper bound for $\delta$. First, observe that by a parity argument, we have that $4\delta$ is bounded by $2\mathsf{d_{T}}(a,b)-2$ instead of $2\mathsf{d_{T}}(a,b)-1$. Secondly, since $\delta$ is a non-negative integer, we actually have $4\delta \leq \lfloor 2\mathsf{d_{T}}(a,b)-2 \rfloor$. Therefore, we obtain 

\begin{align*}
\wtr(\e_{i+1})+P_{i+1} &= 2\mathsf{d_{T}}(a,b)-1+1-\delta \geq 2\mathsf{d_{T}}(a,b) - \lfloor \frac{\mathsf{d_{T}}(a,b)-1}{2} \rfloor \\
&= 2\wtr(\e_{i}) - \lfloor \frac{\wtr(\e_{i})-1}{2} \rfloor \geq \frac{5}{4}(\wtr(\e_{i}) + 1) \geq \frac{6}{5}(\wtr(\e_{i}) + 1).
\end{align*}
Indeed, $\wtr(\e_{i}) \mapsto (2\wtr(\e_{i}) - \lfloor \frac{\wtr(\e_{i})-1}{2} \rfloor)(\wtr(\e_{i}) + 1)^{-1}$ is minimal for $\wtr(\e_{i})=3$ and equals $\frac{5}{4}$. 

Case $\mu = -2$: we have moved the endpoints of $\mathbf{p}$ closer together such that $\wtr(\e_{i+1})$ is decreased by $2$. Note that this requires that $\wtr(\e_{i})$ is at least $3$. To achieve the desired bound, one should again be  careful when bounding $\delta$. To avoid edges that are going in opposite directions, we cannot remove a double-edge near one of the extremities of the path~$\mathbf{p}$. It follows that $4\delta$ is bounded above by $\lfloor 2\mathsf{d_{T}}(a,b)-4 \rfloor$. Hence, 
\begin{align*}
\wtr(\e_{i+1})+P_{i+1} &= 2\mathsf{d_{T}}(a,b)-2+1 -\delta \geq 2\mathsf{d_{T}}(a,b) - 1 -\lfloor \frac{\mathsf{d_{T}}(a,b)}{2} \rfloor +1 \\
&= 2\wtr(\e_{i}) - \lfloor \frac{\wtr(\e_{i})}{2} \rfloor \geq \frac{6}{5}(\wtr(\e_{i}) + 1).
\end{align*}
Indeed, $\wtr(\e_{i}) \mapsto (2\wtr(\e_{i}) - \lfloor \frac{\wtr(\e_{i})}{2} \rfloor)(\wtr(\e_{i}) + 1)^{-1}$ is minimal for $\wtr(\e_{i})=4$ and equals $\frac{6}{5}$.
\end{proof}

Note that the inequality in Lemma~\ref{lem:growth 2D} cannot be improved upon since it is actually possible to find error patterns $\e_{i}$ and their preimages such that equality is achieved. Two examples are shown in Figure~\ref{fig:examples satisfying growth 2D}.

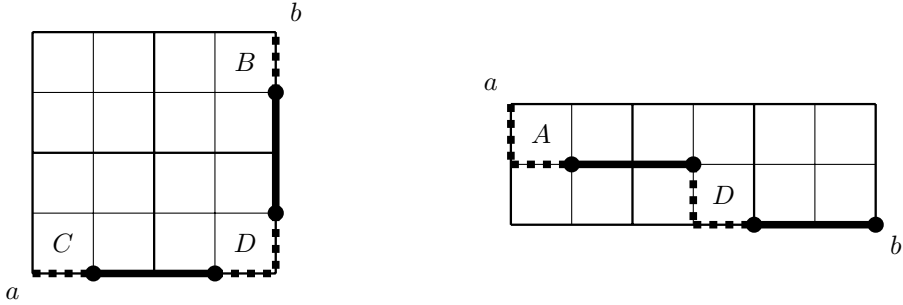
\begin{figure}[htbp]
\centering
\subfloat{
\begin{tikzpicture}[scale=0.8]
\draw[step=1] (0,0) grid (4,4);
\draw[step=2, line width=0.3mm] (0,0) grid (4,4);
\draw[line width=1mm] (1,0) -- (3,0);
\draw[line width=1mm] (4,1) -- (4,3);
\draw[line width=1mm, dashed] (4,3) -- (4,4);
\draw[line width=1mm, dashed] (0,0) -- (1,0);
\draw[line width=1mm, dashed] (3,0) -- (4,0);
\draw[line width=1mm, dashed] (4,1) -- (4,0);
\draw[line width=1mm] (1,0) circle (2pt);
\draw[line width=1mm] (3,0) circle (2pt);
\draw[line width=1mm] (4,1) circle (2pt);
\draw[line width=1mm] (4,3) circle (2pt);
\node at (0.5,0.5) {$C$};
\node at (3.5,3.5) {$B$};
\node at (3.5,0.5) {$D$};
\node at (-0.33,-0.33) {$a$};
\node at (4.33,4.33) {$b$};
\end{tikzpicture}}
\hfil
\raisebox{0.6cm}{
\subfloat{
\begin{tikzpicture}[scale=0.8]
\draw[step=1] (0,0) grid (6,2);
\draw[step=2, line width=0.3mm] (0,0) grid (6,2);
\draw[line width=1mm] (1,1) -- (3,1);
\draw[line width=1mm] (4,0) -- (6,0);
\draw[line width=1mm, dashed] (0,2) -- (0,1);
\draw[line width=1mm, dashed] (0,1) -- (1,1);
\draw[line width=1mm, dashed] (3,1) -- (3,0);
\draw[line width=1mm, dashed] (3,0) -- (4,0);
\draw[line width=1mm] (1,1) circle (2pt);
\draw[line width=1mm] (3,1) circle (2pt);
\draw[line width=1mm] (4,0) circle (2pt);
\draw[line width=1mm] (6,0) circle (2pt);
\node at (0.5,1.5) {$A$};
\node at (3.5,0.5) {$D$};
\node at (-0.33,2.33) {$a$};
\node at (6.33,-0.33) {$b$};
\end{tikzpicture}}}
\caption{Two error patterns (dashed and thick edges) and their preimages (thick edges) such that, in both cases, we have $\wtr(\e_{i}) + P_i = 4 + 1 = 5$ and $\wtr(\e_{i+1})+P_{i+1} = 4 + 2 = 6$. }
\label{fig:examples satisfying growth 2D}
\end{figure}

We now prove that the error-correcting radius $\omega$ is bounded below in terms of $d^{\log_{2}(6/5)}$. 

\begin{proof}[Proof of Theorem~\ref{thm:lower bound}]
We need to show that if an error pattern $\e_{k}$ in the graph $\bT_{k}$ is decoded incorrectly, then its weight satisfies $\wt(\e_{k}) \geq v_{k}$, where $(v_{k})_{k \in \N}$ is the sequence defined by $v_k=\left( \frac{6}{5} \right)^{k-1}$. 

By exhausting all the possible error patterns, one can show that for $k = 1$, the minimal weight of an error pattern in $\bT_{1}$ that is decoded incorrectly is equal to $v_{1}$. Since $\e_{1}$ is decoded incorrectly, we thus have that $\wt(\e_{1}) \geq v_{1}$. Suppose now that $k \geq 2$. By Lemma~\ref{lem:growth 2D}, we have that $\forall i \in [[1,k]]$, $\wtr(\e_{i+1}) + P_{i+1} \geq \frac{6}{5} (\wtr(\e_{i}) + P_{i})$. Hence by repeatedly applying Lemma~\ref{lem:growth 2D}, we obtain \[\wtr(\e_{k}) + P_{k} \geq \left( \frac{6}{5} \right)^{k-1} (\wtr(\e_{1}) + P_{1}) \geq~2\left( \frac{6}{5} \right)^{k-1}.\] Since each path has at least one edge, it follows that the Hamming weight of the error $\e_k$ is at least as large as its number of paths, $\wt(\e_{k}) \geq P_{k}$. Moreover, the Hamming weight of the error $\e_{k}$ is greater or equal than its reduced weight, $\wt(\e_{k}) \geq \wtr(\e_{k})$. Hence, we have that $\wt(\e_{k}) \geq \frac{\wtr(\e_{k}) + P_{k}}{2} \geq \left( \frac{6}{5} \right)^{k-1} = v_{k}$.
\end{proof}

\section{Fractal-like errors for more general renormalisation decoders}
\label{sec:fractal errors}

The analysis of the error-correcting radius of the renormalisation decoder depends of course on the precise instructions defined in Section~\ref{subsec:precise description}. We would like to stress that the existence of fractal-like wrongly decoded sublinear error patterns similar to that of Section~\ref{subsec:wrongly decoded error} will be a feature of any deterministic renormalisation decoder and not just of our particular decoder choice. In this section, we consider more general deterministic and hard-decision renormalisation decoders and show they suffer from wrongly decoded error patterns of weight $d^\alpha$ inside minimal non-trivial cycles.

\subsection{Renormalisation decoders with larger blocks}
\label{subsec:errors for larger blocks}

A first approach to generalise the decoder we introduced in Section~\ref{sec:renormalisation decoder}, is to subdivide the toric graphs with larger blocks. If $b \geq 2$ denotes the size of the blocks of the code, then the code has length $n = 2 b^{2k}$ and distance $d = b^k$. 

We could define new sets of instructions for these decoders but to illustrate our point, we may just ask the following natural property. If a single vertex syndrome is present inside a block, then the decoder shifts it to the closest corner vertex of the block. 

We show the existence of fractal-like wrongly decoded errors for $b = 3$. The first $3$ steps of a wrongly decoded error inside a minimal non-trivial cycle are shown in Figure~\ref{fig:wrongly decoded error larger blocks}. One constructs the following errors $\e_i$ for $i \geq 4$ by repeating the same method that is used to obtain $\e_3$ from $\e_2$ to each path of $\e_{i-1}$. By computing the weight of the errors, we then have $\wt(\e_k) = 2^{k} = d^{\, \log_3(2)} \simeq d^{\, 0.63}$, where $d = 3^{k}$.

\begin{figure}[htbp]
\centering
\begin{tikzpicture}[scale=1.05]
\foreach \i in {0,...,3}{
	\node at (0,\i /2) {{\footnotesize $i=\i$}};
    \draw[line width=0.25mm] (1,\i /2)  -- (14.5,\i /2);
}

\foreach \j in {1,14.5}{
    		\draw[line width=0.25mm] (\j,0) circle (1pt);
}
\foreach \j in {1,5.5,10,14.5}{
    		\draw[line width=0.25mm] (\j,0.5) circle (1pt);
}
\foreach \j in {1,2.5,4,...,14.5}{
    		\draw[line width=0.25mm] (\j,1) circle (1pt);
}
\foreach \j in {1,1.5,...,14.5}{
    		\draw[line width=0.25mm] (\j,1.5) circle (1pt);
}
\draw[line width=0.5mm] (1,0) -- (14.5,0);
\draw[line width=0.5mm] (1,0.5) -- (10,0.5);
\draw[line width=0.5mm] (2.5,1) -- (8.5,1);
\draw[line width=0.5mm] (3,1.5) -- (5,1.5);
\draw[line width=0.5mm] (6,1.5) -- (8,1.5);
\draw[line width=0.5mm] (1,0) circle (2pt);
\draw[line width=0.5mm] (14.5,0) circle (2pt);
\draw[line width=0.5mm] (1,0.5) circle (2pt);
\draw[line width=0.5mm] (10,0.5) circle (2pt);
\draw[line width=0.5mm] (2.5,1) circle (2pt);
\draw[line width=0.5mm] (8.5,1) circle (2pt);
\draw[line width=0.5mm] (3,1.5) circle (2pt);
\draw[line width=0.5mm] (5,1.5) circle (2pt);
\draw[line width=0.5mm] (6,1.5) circle (2pt);
\draw[line width=0.5mm] (8,1.5) circle (2pt);
\end{tikzpicture}
\caption{Construction of the wrongly decoded error pattern for a renormalisation decoder with block length $b = 3$. For $i \in [[0,3]]$, the error $\e_{i}$ is shown in the first horizontal line of the graph $\bT_i$.}
\label{fig:wrongly decoded error larger blocks}
\end{figure}
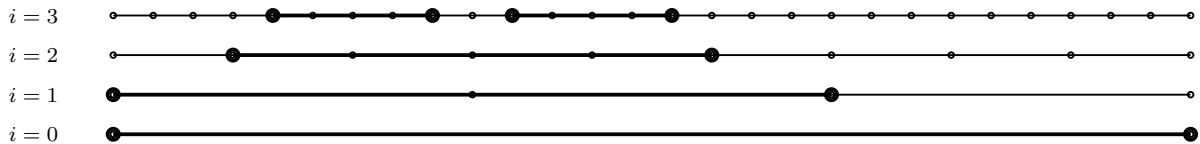 

More generally, for a renormalisation decoder with blocks of size $b$, there exists a fractal-like error pattern of weight $\wt(\e_{k}) = \lfloor \frac{b}{2}+1 \rfloor^{k} = d^{\, \log_b(\lfloor \frac{b}{2}+1 \rfloor)}$, where $d = b^{k}$. Note that for large values of $b$, the result isn't really meaningful since the only condition we assumed on the decoders is weak. Indeed, it is restrictive to assume that only a single vertex syndrome can be present in a block since this means long chains of error edges.

\subsection{Renormalisation decoders with message-passing}
\label{subsec:errors for message-passing}

A second approach to generalise the decoder we introduced in Section~\ref{sec:renormalisation decoder}, is to add message-passing. This means that the decoder will not only base its decisions on a block in the reduction procedure, but it will also consider the blocks around it. If $m$ denotes the depth of the message-passing, then we make our decisions considering all the blocks in a radius~$m$. 

We could define new sets of instructions to explain how the message-passing aids the decoders but to illustrate our point, we again just ask the following natural property. We may naturally assume that if a single vertex syndrome is present inside a radius of $m$ blocks, then the decoder shifts it to the closest corner vertex of the block. 

Let us start to show the existence of fractal-like wrongly decoded errors for $m = 2$. The first $5$ steps of a wrongly decoded error pattern inside a minimal non-trivial cycle are shown in Figure~\ref{fig:wrongly decoded error message-passing}. One constructs the following errors $\e_i$ for $i \geq 6$ by repeating the same method that is used to obtain $\e_4$ from $\e_3$ to each path of the error $\e_{i-1}$ of weight $4$, and the same method that used to obtain $\e_3$ from $\e_2$ to each path of the error $\e_{i-1}$ of weight $2$. We note that the paths of length $2$ have a preimage consisting of one path of length $4$, and the paths of length $4$ have a preimage consisting of one path of length $2$ and one of length $4$. By computing the weight of the errors, we then have for $k \geq 2$, $\wt(\e_k) = 2 F_k$, where $(F_k)_{k \in \mathbb{N}}$ is the Fibonacci sequence. Since for $k \geq 1$, $F_k \leq \phi^{k-1}$, where $\phi = \frac{1+\sqrt{5}}{2}$, we have that $\wt(\e_k) \leq \frac{2}{\phi} \phi^{k} = \frac{2}{\phi} d^{\, \log_2(\phi)} \simeq \frac{2}{\phi}  d^{\, 0.69}$.

\begin{figure}[htbp]
\centering
\begin{tikzpicture}[scale=0.9]
\foreach \i in {0,...,5}{
	\node at (0,\i /2) {{\footnotesize $i=\i$}};
    \draw[line width=0.25mm] (1,\i /2)  -- (17,\i /2);
}

\foreach \j in {1,17}{
    		\draw[line width=0.25mm] (\j,0) circle (1pt);
}
\foreach \j in {1,9,17}{
    		\draw[line width=0.25mm] (\j,0.5) circle (1pt);
}
\foreach \j in {1,5,...,17}{
    		\draw[line width=0.25mm] (\j,1) circle (1pt);
}
\foreach \j in {1,3,...,17}{
    		\draw[line width=0.25mm] (\j,1.5) circle (1pt);
}
\foreach \j in {1,...,17}{
    		\draw[line width=0.25mm] (\j,2) circle (1pt);
}
\foreach \j in {1,1.5,...,17}{
    		\draw[line width=0.25mm] (\j,2.5) circle (1pt);
}
\draw[line width=0.5mm] (1,0) -- (17,0);
\draw[line width=0.5mm] (1,0.5) -- (9,0.5);
\draw[line width=0.5mm] (1,1) -- (9,1);
\draw[line width=0.5mm] (1,1.5) -- (9,1.5);
\draw[line width=0.5mm] (1,2) -- (3,2);
\draw[line width=0.5mm] (4,2) -- (8,2);
\draw[line width=0.5mm] (1,2.5) -- (3,2.5);
\draw[line width=0.5mm] (4,2.5) -- (5,2.5);
\draw[line width=0.5mm] (5.5,2.5) -- (7.5,2.5);
\draw[line width=0.5mm] (1,0) circle (2pt);
\draw[line width=0.5mm] (17,0) circle (2pt);
\draw[line width=0.5mm] (1,0.5) circle (2pt);
\draw[line width=0.5mm] (9,0.5) circle (2pt);
\draw[line width=0.5mm] (1,1) circle (2pt);
\draw[line width=0.5mm] (9,1) circle (2pt);
\draw[line width=0.5mm] (1,1.5) circle (2pt);
\draw[line width=0.5mm] (9,1.5) circle (2pt);
\draw[line width=0.5mm] (1,2) circle (2pt);
\draw[line width=0.5mm] (3,2) circle (2pt);
\draw[line width=0.5mm] (4,2) circle (2pt);
\draw[line width=0.5mm] (8,2) circle (2pt);
\draw[line width=0.5mm] (1,2.5) circle (2pt);
\draw[line width=0.5mm] (3,2.5) circle (2pt);
\draw[line width=0.5mm] (4,2.5) circle (2pt);
\draw[line width=0.5mm] (5,2.5) circle (2pt);
\draw[line width=0.5mm] (5.5,2.5) circle (2pt);
\draw[line width=0.5mm] (7.5,2.5) circle (2pt);
\end{tikzpicture}
\caption{Construction of the wrongly decoded error pattern for a renormalisation decoder with message-passing of depth $m = 1$. For $i \in [[0,5]]$, the error $\e_{i}$ is shown in the first horizontal line of the graph $\bT_i$.}
\label{fig:wrongly decoded error message-passing}
\end{figure}
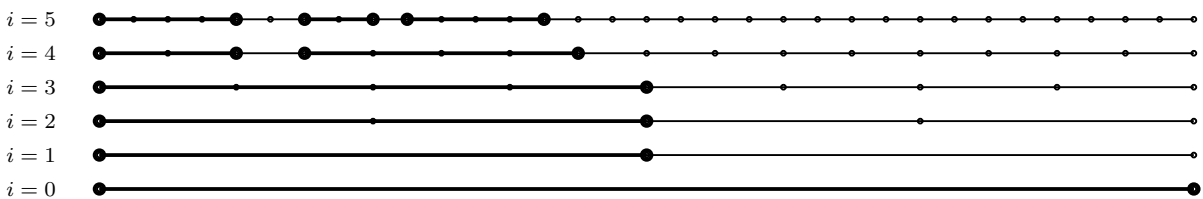

More generally, it is possible to find fractal-like wrongly decoded error patterns of weight $d^\alpha$, even if improving the depth of the message-passing will increase the exponent~$\alpha$.

\section{Concluding comments}
\label{sec:conclusion}
We conjecture that the error-correcting radius $\omega$ scales as $d^{1/2}$. However, the natural expectation that minimal wrongly decoded errors are purely $1$-dimensional is false. It is possible to find $2$-dimensional wrongly decoded error patterns that have smaller weight than the ones introduced in Section~\ref{subsec:wrongly decoded error}. Nevertheless, the weight of these error patterns still scales as $d^{1/2}$. The situation is therefore not so clear-cut.

\section*{Acknowledgements}
We acknowledge support from the Plan France $2030$ through the project NISQ$2$LSQ, ANR-$22$-PETQ-$0006$. GZ acknowledges support from the ANR through the project QUDATA, ANR-$18$-CE$47$-$0010$.

%% References:

\end{document}